%% file: paper-combined.tex
\documentclass[letterpaper, 10 pt,twoside]{IEEEtran}  

\usepackage{tablefootnote}
\usepackage{amsthm}
\usepackage{amssymb}
\usepackage{amsfonts}
\usepackage{amsmath}
\usepackage{graphics}
\usepackage{graphicx}
\usepackage{epsfig}
\usepackage{epstopdf}
\usepackage{mathrsfs}
\usepackage{cite}
\usepackage[mathscr]{euscript}

\include{philipMacros}

\newcommand{\overbar}[1]{\mkern 1.5mu\overline{\mkern-1.5mu#1\mkern-5mu}\mkern 1.5mu}

\newtheorem{theorem}{Theorem}

\newtheorem{lemma}{Lemma}[theorem]
\newtheorem{claim}[lemma]{Claim}
\newtheorem{corollary}[theorem]{Corollary}

\newtheorem{definition}{Definition}

\newtheorem{example}{Example}

\begin{document}
%
\title{Can Taxes Improve Congestion on All Networks?}

\author{{Philip N. Brown and Jason R. Marden}
\thanks{This work is supported ONR Grant \#N00014-17-1-2060 and NSF Grant \#ECCS-1638214.}
\thanks{P. N. Brown is with the Dept. of Computer Science, University of Colorado at Colorado Springs, {\texttt{philip.brown@uccs.edu}. Corresponding author.}}
\thanks{J. R. Marden is with the Department of Electrical and Computer Engineering, University of California, Santa Barbara, {\texttt{jrmarden@ece.ucsb.edu}.  }}
\thanks{This article draws from conference papers that appeared in~\cite{Brown2016d,Brown2017b,Brown2017e}.}
}


%


\maketitle

\begin{abstract}
We ask if it is possible to positively influence social behavior with no risk of unintentionally incentivizing pathological behavior.
In network routing problems, if network traffic is composed of many individual agents, it is known that self-interested behavior among the agents can lead to suboptimal network congestion.
We study situations in which a system planner charges monetary tolls for the use of network links in an effort to incentivize efficient routing choices by the users, but in which the users' sensitivity to tolls is heterogeneous and unknown.
We seek locally-computed tolls that are guaranteed not to incentivize worse network routing than in the un-influenced case.
Our main result is to show that if networks are sufficiently complex and populations sufficiently diverse, perverse incentives cannot be systematically avoided: any taxation mechanism that improves outcomes on one network must necessarily degrade them on another.
Nonetheless, for the simple class of parallel networks, non-perverse taxes do exist; we fully characterize all such taxation mechanisms, showing that they are a generalized version of traditional marginal-cost tolls.
\end{abstract}


%
\IEEEpeerreviewmaketitle

%

\section{Introduction}

Modern computational and infrastructure systems are becoming increasingly interconnected with the social systems that they serve.
Accordingly, engineers must be aware of the ways in which social behavior affects system performance; this has spurred recent research on influencing social behavior to achieve engineering objectives~\cite{Brown2017d,Alizadeh2014,Gopalakrishnan2014}.
Examples of problems in this context include ridesharing systems~\cite{Kleiner2011}, transportation networks~\cite{Brown2017a,Brown2016d,Brown2017b,Brown2017e}, and power grids~\cite{Silva2001}.

In this paper, we study a popular model of traffic congestion known as ``non-atomic congestion games,'' in which traffic needs to be routed across a network from a source node to a destination node in a way that minimizes the average delay experienced by the traffic.
If a central authority can control the traffic explicitly, it is typically straightforward to compute the optimal assignment of traffic; unfortunately, if the mass of traffic is composed of individual decision-makers, the aggregate network flows that emerge from individual self-interested decision-making may be far from optimal.
This inefficiency due to self-interested decisions is termed the \emph{price of anarchy}, formally defined as the ratio between the total congestion of a selfishly-routed flow and that of an optimal flow, taken in worst case over a class of games~\cite{Papadimitriou2001}.
It is known that even on two-link parallel networks with convex latency functions, the price of anarchy is unbounded -- that is, selfish flows can be arbitrarily inefficient~\cite{Roughgarden2002}.

Accordingly, much research has focused on methods of influencing the routing choices made by individual users as a means to improve the price of anarchy; one promising set of methodologies involves charging specially designed tolls to users of network links in an effort to incentivize more-efficient network flows~\cite{Fleischer2004,Karakostas2004}.
In~\cite{Beckmann1956,Sandholm2002} it is shown that if a special type of tolling function called a \emph{marginal-cost} toll is levied on each network link, that this incentivizes optimal network routing -- provided that all network users trade off time and money equally.
An attractive feature of marginal-cost tolls is that they can be computed locally on each network link: a link's toll depends \emph{only} on that link's congestion characteristics and traffic flow.
Thus, the optimality guaranteed by these tolls is intrinsically robust to variations or mischaracterizations of network structure.
This local-computation property is known as \emph{network agnosticity}; in essence, marginal-cost tolling functions only ``know" their own edge -- they are agnostic to global network specifications~\cite{Brown2017d}.

The optimality guarantees coupled with the network agnosticity of marginal-cost tolls suggest the appealing possibility that a system operator could design link prices to incentivize efficient routing without needing to know the overall structure of the network. 
One benefit of an effective network-agnostic taxation mechanism is that it would be largely robust to sudden ``changes'' in network topology such as those caused by weather events or construction projects.

Unfortunately, it has recently been shown that if the user population is diverse in price-sensitivity, the optimality guarantees of marginal-cost tolls vanish~\cite{Brown2017a}.
That is, if some users value their time more than others, networks exist on which the routing incentivized by marginal-cost tolls has higher congestion than un-influenced routing.
Furthermore, if the prices are constrained to be fixed (i.e., constant functions of link flow rate), \emph{every} network agnostic taxation mechanism can create perverse incentives on some parallel network -- even for a homogeneous user population.
Whether {flow-varying} taxation functions suffer from the same ubiquitous perversity has remained an open question.

To study this phenomenon rigorously, we formulate a new robustness metric that we term the \emph{perversity index}.
We define the perversity index of a taxation mechanism as the ratio between the total congestion it incentivizes and the un-influenced congestion, taken in worst case over user populations and networks.
That is, if a taxation mechanism has a perversity index strictly greater than $1$, this indicates that networks exist on which this mechanism degrades the quality of flows, rather than improving it.
When this is the case, we say that a taxation mechanism is \emph{perverse}.
We ask the following questions:
\begin{enumerate}
\item Do there exist network-agnostic taxation mechanisms which can improve worst-case congestion on the class of all networks while avoiding causing congestion degradations on each network instance?
\item What are the common characteristics of taxation mechanisms which systematically avoid perverse incentives?
\end{enumerate}


In this paper, we first show definitively in Theorem~\ref{thm:neg} that if networks are sufficiently complex and the user population is sufficiently diverse, every network-agnostic taxation mechanism that improves outcomes on some networks must degrade outcomes on others.
That is, every network-agnostic taxation mechanism is either {trivial} or {perverse.} 
Let $\poa\left(\gee,T\right)$ and $\perv\left(\gee,T\right)$ denote the price of anarchy and perversity index (respectively) of class of routing problems $\gee$ under the influence of network-agnostic taxation mechanism $T$, and let $\poa\left(\gee,\emptyset\right)$ denote the price of anarchy of uninfluenced problems (these notations are formally defined later in Section~\ref{ssec:tmpm}).
For sufficiently-rich $\gee$, Theorem~\ref{thm:neg} states that
\begin{equation}
\poa\left(\gee,T\right)<\poa\left(\gee,\emptyset\right) \ \implies \ \perv\left(\gee,T\right)>1.
\end{equation}

Since Theorem~\ref{thm:neg} rules out the possibility of a non-perverse mechanism for general routing problems, we next turn to a study of the relationship between PoA and PI on smaller classes of routing problems.
We consider the case of \emph{heterogeneous} populations on the class of parallel-path networks.
Here, we fully characterize the set of all network-agnostic taxation mechanisms having a perversity index of $1$, and show that they are all a simple variant of classical marginal-cost tolls. 
Among these, we also derive tolls which minimize the price of anarchy, and show that these are equivalent to marginal-cost tolls which are ``tuned'' to the most-sensitive users.

Finally, considering the simplified case of homogeneous populations on general networks, we show in Theorem~\ref{thm:opt} that tolls which optimize the price of anarchy are in some sense \emph{maximally perverse.}
That is, these tolls render the price of anarchy equal to the perversity index, meaning that under the influence of these optimal tolls, the problem instances on which worst-case routing outcomes occur would experience lower congestion without tolls.
%
%
We point out that a consequence of this is a simple procedure for minimizing the price of anarchy subject to an upper bound on perversity.

We expect that the perversity index will provide a rich framework for studying many different settings in which a planner desires to positively influence social behavior.
While this paper studies the perversity index for congestion game pricing problems, such a measure of the risk of unintended consequences may be of interest to a political organization which wishes to assure that a proposed policy will not unduly disenfranchise certain groups.

\section{Model and Related Work}

\subsection{Routing Game}

Consider a network routing problem for a network $(V,E)$ comprised of vertex set $V$ and edge set $E$.
There is a mass of traffic $r>0$ that needs to be routed from a common source $\sigma$ to a common destination $t$. 
We write $\pset\subset 2^E$ to denote the set of \emph{paths} available to traffic, where each path $p\in\pset$ consists of a set of edges connecting $\sigma$ to $t$. 
A network is called a \emph{parallel} network if all paths are disjoint; i.e., for all paths $p,p'\in\pset$, $p\cap p'=\emptyset$.
%
%
We write $f_p\geq 0$ to denote the mass of traffic using path $p$. 
A \emph{feasible flow} $f\in\mathbb{R}^{|\pset|}$ is an assignment of traffic to various paths such that $\sum_{p\in\pset}f_p = r$. 

Given a flow $f$, the flow on edge $e$ is given by $f_e = \sum_{p:e\in p}f_p$.
To characterize transit delay as a function of traffic flow, each edge $e\in E$ is associated with a specific latency function $\ell_e:[0,r]\rightarrow[0,\infty)$; $\ell_e(f_e)$ denotes the delay experienced by users of edge $e$ when the edge flow is $f_e$.
We adopt the standard assumptions that each latency function is nondecreasing, convex, and continuously differentiable.
We measure the cost of a flow $f$ by the \emph{total latency}, given by
\vs
\begin{equation}
{\cal L}(f) =  \sum_{e\in E}  f_e \cdot \ell_e(f_e)=  \sum\limits_{p\in \pset}  f_p \cdot \ell_p(f), \label{eq:totlatfpath}
\vs
\end{equation}
where $\ell_p(f) = \sum_{e\in p}\ell_e(f_e)$ denotes the latency on path $p$.
We denote the flow that minimizes the total latency by
\vs
\begin{equation} 
f^* \in \underset{f {\rm \ is \ feasible}}{\rm arg min} \ {\cal L}(f). 
\vs
\end{equation}

A \emph{routing problem} is given by $G=\left(V,E,\left\{\ell_e\right\}\right)$.
We denote classes of routing problems with the calligraphic $\gee$.

To study the effect of taxes on self-interested behavior, we model the above routing problem as a non-atomic congestion game.
We assign each edge $e \in E$ a flow-dependent  taxation function $\tau_e : \mathbb{R}^+ \rightarrow \mathbb{R}$.
We characterize the taxation sensitivities of the users with a monotone, nondecreasing function $s:[0,r]\rightarrow [\bmin,\bmax]$, where each user $x \in [0,r]$ has a taxation sensitivity $s_x \in [S_{\rm L}, S_{\rm U}]\subseteq \mathbb{R}^+$, where $\bmin\geq0$ and $\bmax\leq+\infty$ are lower and upper sensitivity bounds, respectively.
To avoid trivialities, we assume that zero measure of traffic has sensitivity exactly equal to $0$, or that for all $\epsilon>0$, $s_\epsilon>0$. %
If all users have the same sensitivity (i.e., $s_x=s_y$ for all $x\in[0,r]$ and $y\in[0,r]$), the population is said to be \emph{homogeneous}; otherwise it is \emph{heterogeneous}.

Given a flow $f$, the cost that user $x\in[0,r]$ experiences for using path ${p} \in \pset$ is of the form%
%
%
\begin{equation}
J_x(p;f) = \sum\limits_{e\in {p}}\left[\ell_{e}(f_{e}) + s_x  \tau_{e}(f_{e})\right].
\end{equation}
Thus, for each user $x\in[0,r]$, the sensitivity $s_x$ can be viewed as a constant gain on the toll; a user's experienced cost is then the sum of the latency and sensitivity-weighted toll.
Note that sensitivity can be interpreted as the reciprocal of an agent's value-of-time. %
We assume that each user selects the lowest-cost path from the available source-destination paths.
We call a flow $f$ a \emph{Nash flow} if all users are individually using minimum-cost paths given the choices of other users.
That is, every user $x \in [0,r]$ using path $p$ in $f$ experiences a cost satisfying
\begin{equation}
J_x(p;f) = \min_{\tilde{p} \in \pset} J_x \left(\tilde{p};f\right).
\end{equation}
%
It is well-known that a Nash flow exists for any non-atomic congestion game of the above form \cite{Mas-Colell1984}.  

In our analysis, we assume that each sensitivity distribution function $s$ is unknown to the pricing authority; for a given routing problem $G$ and $\bmax\geq\bmin\geq0$ we define the set of possible sensitivity distributions as the set of monotone, nondecreasing functions
$
\see_G = \{ s:[0,r] \rightarrow [S_{\rm L}, S_{\rm U}] \}.
$ 
We write $s\in\see_G$ to denote a specific collection of sensitivity distributions, which we term a \emph{population}.

\subsection{Taxation Mechanisms and Performance Metrics} \label{ssec:tmpm}

To directly study the importance of knowing network structure on the effectiveness of taxation mechanisms, we consider \emph{network-agnostic} taxation mechanisms as in~\cite{Brown2017a}.
Here, each edge's taxation function is computed using only locally-available information.
That is, $\tau_e(f_e)$ depends only on $\ell_e$, not on edge $e$'s location in the network, the network topology, the overall traffic rate, or the properties of any other edge.
A network-agnostic taxation mechanism $\tmech$ is thus a mapping from latency functions to taxation functions, and the taxation function associated with latency function $\ell_e$ is given by
\vs
\begin{equation}
	\tau_e(\cdot) = \tmech(\ell_e).
\vs
\end{equation}

To evaluate the performance of taxation mechanisms, we write $\Lnf(G,s,\tmech)$ to denote the total latency of a Nash flow for routing problem $G$ and population $s$ induced by taxation mechanism $\tmech$.
Let  $\Lnf(G,\emptyset)$ denote the total latency of an un-influenced Nash flow on $G$, and let $\Lopt(G)$ denote the total latency of the optimal flow on $G$.
The \emph{price of anarchy} compares the Nash flows induced by taxes with the optimal flows in worst case over routing problems and populations.
Formally, the price of anarchy of a class of games $\gee$ under the influence of taxes $\tmech$ is defined as
\begin{equation} \label{eq:poadef}
\poa\left(\gee,\tmech\right) \triangleq \sup_{G\in\gee}\sup_{s\in\see_G} \frac{\Lnf(G,s,\tmech)}{\Lopt(G)}.
\end{equation}

In this paper, we pose a somewhat different question: rather than measuring how far the influenced flows are from \emph{optimal}, it may be desirable to quantify the risk of causing harm to congestion \emph{with respect to the un-influenced flows.}
This can be easily captured with a natural modification of the price of anarchy concept by replacing the optimal total latency with the un-influenced total latency $\Lnf(G,\emptyset)$.
We call such a metric the \emph{index of perversity} of taxation mechanism $\tmech$, defined as
\begin{equation} \label{eq:pidef}
\perv\left(\gee,\tmech\right) \triangleq \sup_{G\in\gee}\sup_{s\in\see_G} \frac{\Lnf(G,s,\tmech)}{\Lnf(G,\emptyset)}.
\end{equation}

Here, if $\tmech$ has a large index of perversity, this means that on some networks, it incentivizes flows that are considerably worse than the un-influenced flows; in other words, it can create perverse incentives.
If a taxation mechanism has a perversity index of~$1$, we say that it is \emph{non-perverse}; otherwise, it is \emph{perverse}.
Note that it is always true that $\perv(\gee,\tmech)\leq\poa(\gee,\tmech)$; this is because on any $G$, $\Lnf(G,\emptyset)\geq\Lopt(G)$.
Finally, when these metrics are evaluated only over homogeneous populations (as opposed to heterogeneous), we write them as $\poa^{\rm hm}(\gee,\tmech)$ and $\perv^{\rm hm}(\gee,\tmech)$, respectively.

\subsection{Related Work} \label{ssec:background}

The following is a brief survey of relevant work on the robustness and perversity of taxation mechanisms in congestion games. 
Much work focuses on \emph{network-dependent} taxation mechanisms, in which each edge toll is a function of the entire routing problem.
Network dependency allows taxation mechanisms to incentivize precisely-targeted network flows, as with 
\emph{fixed tolls}, which for any $e \in G$, $\tau_e(f_e) = q_e$ for some $q_e \geq 0$.
If network structure, traffic-rate, and user sensitivity specifications are known precisely, it is possible to compute fixed tolls which induce optimal Nash flows on any network~\cite{Fleischer2004,Karakostas2004}.
However, if any of these pieces of information are unknown, no fixed tolls exist which guarantee optimal routing, and the perversity index of fixed tolls is typically greater than 1~\cite{Brown2017a}.
Network-dependent tolls have been studied subject to many restrictions: bounded, budget-balanced, and dynamic schemes have been proposed~\cite{Bonifaci2011,Jelinek2014,Arieli2015,Bhaskar2014}.
However, the robustness of most of these schemes to variations of user sensitivity has not been investigated explicitly, and the perversity index in these contexts is generally unknown.

The classical example of a network-agnostic taxation mechanism is that of the \emph{marginal-cost} or \emph{Pigovian} taxation mechanism~$\tmc$. For any edge $e$ with latency function $\ell_e$, the accompanying marginal-cost toll is 
\vs
\begin{equation}
\vs
\tau_e^{\rm mc}\left(f_e\right)  = f_e \cdot \ell_e'(f_e), \ \forall f_e \geq 0, \label{eq:mc}
\end{equation}
where $\ell'$ represents the flow derivative of $\ell$.
It has long been known that for any $G$, it is true that
$
{\cal L}^*(G) = {\cal L}^{\rm nf}\left(G, s, \tmech^{\rm mc}\right),
$ 
provided that all users have a sensitivity equal to 1~\cite{Beckmann1956}.
That is, for unit-sensitivity homogeneous users, we have a perfect perversity index $\perv\left(\gee,\tmc\right)=\poa \left(\gee,\tmc\right)=1$.

Recent research has identified several new network-agnostic taxation mechanisms, which are all in some sense generalizations of~$\tmc$: in restricted settings, scaled marginal-cost tolls can be non-perverse~\cite{Brown2017c} while guaranteeing improvements; furthermore, there exists a universal taxation mechanism based on~$\tmech^{\rm mc}$ which optimizes routing using large tolls~\cite{Brown2017d}.
Dynamic network-agnostic tolls converging to~$\tmech^{\rm mc}$ are studied in~\cite{Poveda2017}.
In the present paper, Theorem~\ref{thm:nonperv} shows that this  connection to $\tmc$ is no accident; in a sense, the \emph{only} interesting network-agnostic taxation mechanisms are generalizations of~$\tmc$.

Our work here is also tightly connected with the broader theme of ``biased'' congestion games~\cite{Meir2015b}, where players misinterpret edge cost functions in some systematic way.
Several works have investigated the price of anarchy under various payoff biases such as altruism~\cite{Chen2014} and pessimism~\cite{Meir2015b}.
Analogous to our definition of the perversity index, the authors of~\cite{Lianeas2016} study the ``price of risk aversion,'' which measures how society's risk preferences affect aggregate congestion as compared to ordinary Nash flows.
Similarly,~\cite{Kleer2017} studies the ``deviation ratio,'' which measures essentially the same quantity for arbitrary cost function biases.

\section{Our Contributions} \label{sec:ourContributions}

Our first question is this: do there exist network-agnostic taxation mechanisms which have a perversity index of $1$?
Example~\ref{ex:braessaug} shows that at least the marginal-cost taxation mechanism~\eqref{eq:mc} has a perversity index strictly greater than $1$; subsequently, Theorem~\ref{thm:neg} shows that this is true for any network-agnostic taxation mechanism.

\begin{example} \label{ex:braessaug}
Consider the network depicted in Figure~\ref{fig:braessaug}, consisting of the well-known Braess's Paradox network~\cite{Braess2005} in parallel with a single constant-latency edge.
Let marginal-cost tolls be charged on the network according to~\eqref{eq:mc}; that is, edges $e_1$ and $e_4$ are each charged a flow-varying toll of $\tau_e(f_e) = f_e$.
If the user population has $2$ units of traffic and a homogeneous toll sensitivity of $s\in[0,1]$, the unique Nash flow on this network is the one labeled ``Uninfluenced/Optimal" in Figure~\ref{fig:braessaug}. Since all agents are experiencing a cost of $2+s$, deviating to the zig-zag path or to $e_6$ would yield a larger cost of $2+2s$ or $3$, respectively.
Since there are $2$ units of traffic experiencing a delay of $2$ each, the total latency is $4$.

Now consider a heterogeneous population in which $1$ unit of traffic has a sensitivity of $s_1=0$ (the orange traffic in Figure~\ref{fig:braessaug}), and $1$ unit of traffic has a sensitivity of $s_2=1$.
In this case, a new Nash flow emerges: one in which all the insensitive traffic uses the zig-zag path, and all the sensitive traffic uses the constant-latency link, labeled ``Marginal-cost Tolls" in Figure~\ref{fig:braessaug}.
In this flow, any agent on the zig-zag path has a delay of $2$, but any agent on the constant-latency path has a delay of $3$, for a total latency of $2+3 = 5$, which is considerably greater than the un-tolled total latency of $4$.

Note here that we have $s_1=0$ for ease of exposition, but similar results can be obtained for any small positive value for $s_1$, as
we show explicitly in the proof of Theorem~\ref{thm:neg}.

\vspace{-3mm}
\end{example}

\begin{figure}
\vspace{2mm}
\centerline{\includegraphics[scale=0.24]{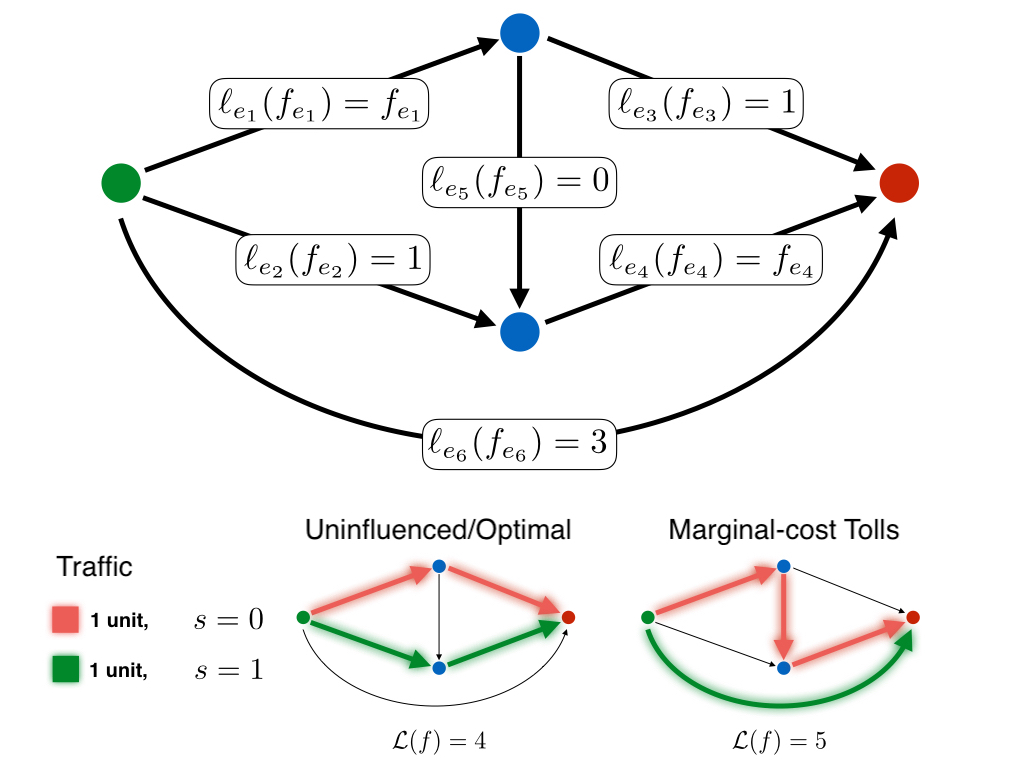}}
\vspace{-1em}
\caption{Example~\ref{ex:braessaug}: A network demonstrating that marginal-cost tolls are perverse on single-source/destination networks.
The user population has mass $r=2$, divided equally between sensitivity types with $s=0$ and $s=1$.
Here, when no tolls are levied, the unique Nash flow for any population is depicted on the left, under the caption ``Uninfluenced/Optimal''; this flow has a total latency of $4$.
However, for the depicted heterogeneous population, marginal-cost tolls induce the Nash flow shown on the right under caption "Marginal-cost Tolls,'' with an increased total latency of $5$.
}
\label{fig:braessaug}
\vspace{-4mm}
\end{figure}

\subsection{Perverse incentives are unavoidable if networks are sufficiently complex}

Our first theorem shows that the pathology shown in Example~\ref{ex:braessaug} is not merely a vestige of the particular form of marginal-cost tolls; indeed, all network-agnostic taxation mechanisms can create perverse incentives if the class of networks is rich enough.

\begin{theorem} \label{thm:neg}
Let $\gee$ denote the class of all routing problems, 
and let $\bmax>\bmin=0$%
\footnote{
Note that $\bmin=0$ is being used in conjunction with the assumption that all traffic has positive sensitivity.
The proof of Theorem~\ref{thm:neg} is completed with all traffic having a strictly positive sensitivity $s_1$, and 
$\bmin=0$ simply ensures that the traffic's sensitivity is not bounded away from zero.
}%
. 
If a network-agnostic taxation mechanism improves the price of anarchy, then it creates perverse incentives.
That is,
\begin{equation}
\poa\left(\gee,T\right)<\poa\left(\gee,\emptyset\right) \ \implies \ \perv\left(\gee,T\right)>1. \label{eq:implication}
\end{equation}
\end{theorem}

A cornerstone of the proof of Theorem~\ref{thm:neg} is the following lemma, which gives a set of necessary conditions for a network-agnostic taxation mechanism to be non-perverse.
The key insight from Lemma~\ref{lem:necessary} is that the only candidate for a non-perverse mechanism is a generalized form of marginal-cost tolls.
Throughout this paper, we write $T(\kone,\ktwo)$ for $\kone,\ktwo\in\mathbb{R}$ to denote a taxation mechanism which assigns tolls of 
\begin{equation}
	\tau_e(f_e) = \kone \ell_e(f_e) + \ktwo f_e \ell_e'(f_e). \label{eq:gmc}
\end{equation}

\begin{lemma} \label{lem:necessary}
Let $\gee$ denote the class of all routing problems.
If network-agnostic taxation mechanism $\tmech$ has $\perv\left(\gee,\tmech\right)=1$, then it is a \emph{generalized marginal-cost taxation mechanism}, written $\tgmc$, which is $T(\kone,\ktwo)$ as in~\eqref{eq:gmc} with
 $\kone>-1/\bmax$, $\ktwo\geq0$, and $\ktwo\leq \kone+1/\bmax$.
\end{lemma}
\vspace{2mm}

\noindent The proof of Lemma~\ref{lem:necessary} appears in the Appendix.

Before continuing, note that the feasible toll-coefficient region specified in Lemma~\ref{lem:necessary} does allow {negative} tolls; in particular, $\ktwo=0$ and $\kone=-1/(2\bmax)$ is always feasible and thus non-perverse.
This can be interpreted as paying users a {subsidy} proportional to the delay on each edge.
However, note that any tolls with $\ktwo=0$ are \emph{trivial}; that is, since they are proportional to delay, they effectively merely scale all delay functions by a constant factor and thus incentivize exactly the un-influenced Nash flows.
Later, Corollary~\ref{cor:poaSL} demonstrates that the price of anarchy is minimized by maximizing $\kone$ and $\ktwo$, thus typically rendering them positive.

The space of network-agnostic taxation mechanisms is quite large, but Lemma~\ref{lem:necessary} dramatically reduces the search, allowing us to search over just two parameters, $\kone$ and $\ktwo$.

\vspace{2mm}
\noindent\emph{Proof of Theorem~\ref{thm:neg}:}

Lemma~\ref{lem:necessary} rules out all taxation mechanisms other than $T(\kone,\ktwo)$ with $\kone> -1/\bmax$, and $\ktwo\leq \kone+1/\bmax$.

First, if $\ktwo=0$, this taxation mechanism assigns taxes of $\tau_e(f_e) = \kone\ell_e(f_e)$ to each edge.
When $\kone>-1/\bmax$, the edge cost functions induced are simply equal to scaled latency functions, meaning that the taxes induce only the uninfluenced Nash flows.
That is, these taxes are trivial: $\ktwo=0$ yields $\Lnf(G,s,\tmech(\kone,0))=\Lnf(G,\emptyset)$ for every network $G$ and population $s$.
For tolls to improve the price of anarchy as in~\eqref{eq:implication}, they must modify some flow on some network; the argument above shows that this requires $\ktwo>0$. Thus, for the remainder of the proof, we assume that $\ktwo>0$.

Our task is to create a user population $s$ (that is, a distribution of tax-sensitivities) and a network $G$ such that $\Lnash(G,s,\tmech(\kone,\ktwo))>\Lnash(G,\emptyset)$.
We will do this with a population having two sensitivity values $s_2>s_1>0$ and a network resembling that in Figure~\ref{fig:braessaug}.
Construct the population as follows: let a unit mass of users have sensitivity $s_1$ (which we will specify momentarily) and a unit mass have $s_2=\bmax$, for a total of 2 units of traffic.
Levy tolls of $T(\kone,\ktwo)$ on the network with $\kone>-1/\bmax$ and $\ktwo>0$.
Define $\gamma_2\triangleq\frac{s_2\ktwo}{1+s_2\kone}\in(0,1]$, and choose $s_1$ so that $\gamma_1\triangleq\frac{s_1\ktwo}{1+s_1\kone}=\gamma_2/8$.
Then an agent with sensitivity $s_i\in\{s_1,s_2\}$ experiences an effective cost function on edge $e$ of
\begin{equation}
J_e(f_e) = \ell_e(f_e) + \gamma_i f_e\ell'_e(f_e).
\end{equation}

Now, let $G$ be the network depicted in Figure~\ref{fig:braessaug}; let $\ell_{e_1}(f_{e_1})=f_{e_1}$, let $\ell_{e_4}(f_{e_4})=f_{e_4}$, let $\ell_{e_2}(f_{e_2})=\ell_{e_3}(f_{e_3})=1+\gamma_2/8$, let $\ell_{e_5}(f_{e_5}) = 0$, and let $\ell_{e_6}\left(f_{e_6}\right)=2+\gamma_2$.
Enumerate the paths as follows: denote the ``zig-zag" path $p_1 = \{e_1,e_5,e_4\}$, the remaining two paths in the upper sub-network $p_2 = \{e_1,e_3\}$ and $p_3 = \{e_2,e_4\}$, and the isolated constant-latency path $p_4 = \{e_6\}$; and denote the path flow of $p_i$ by $f_i$.
We will refer to paths $p_1, p_2,$ and $p_3$ in the upper subnetwork as the ``Braess subnetwork."

On this network, the flow (depicted on the right in Figure~\ref{fig:braessaug}) $\fperv\triangleq(1,0,0,1)$ is a Nash flow for this population with all of population $1$ choosing the zig-zag path and all of population $2$ choosing path $4$.
This is because for population $i\in\{1,2\}$, the effective cost of the four paths under $\fperv$ are $2(1+\gamma_i)$, $(2+\gamma_1+\gamma_i)$,  $(2+\gamma_1+\gamma_i)$, and $(2+\gamma_2)$, respectively.
Since $\gamma_1=\gamma_2/8$, this means that under $\fperv$, population~$1$ weakly prefers $p_1$, and population~$2$ strictly prefers $p_4$.
This flow has total latency $\Lnash(G,s,\tmech)=4+\gamma_2$.

However, it can be verified that if tolls are removed, the unique Nash flow is $\fnash\triangleq(\gamma_2/4,1-\gamma_2/8,1-\gamma_2/8,0)$, which has a total latency of $\Lnash(G,\emptyset)=4+\gamma_2/2$, or
\begin{equation}
\Lnash(G,s,\tmech)>\Lnash\left(G,\emptyset\right)
\end{equation}
and the considered tolls are perverse.
\hfill\qed

\subsection{Parallel networks prevent perverse incentives}

Theorem~\ref{thm:neg} and its proof show that it does not take much complexity to render a network-agnostic taxation mechanism perverse.
Does this mean that it is never possible to achieve a perversity index of $1$ on any class of networks?
Fortunately, the answer is no -- and our Theorem~\ref{thm:nonperv} shows that on parallel networks, %
the necessary condition from Lemma~\ref{lem:necessary} is also sufficient to achieve a perversity index of $1$.
Thus, Theorem~\ref{thm:nonperv} gives a full characterization of non-perverse taxation mechanisms for parallel networks.

\begin{theorem} \label{thm:nonperv}
Let $\gee^{\rm p}$ denote the class of routing problems with parallel networks.
For any $\bmax\geq\bmin\geq0$, a network-agnostic taxation mechanism $\tmech$ has unity perversity index on~$\gee^{\rm p}$, i.e.,
\begin{equation}
\perv\left(\gee^{\rm p},\tmech\right)=1,
\end{equation}
if and only if it is $\tgmc$, assigning the tolling functions
\begin{equation}
	\tau_e(f_e) = \kone \ell_e(f_e) + \ktwo f_e \ell_e'(f_e), \label{eq:nonperv}
\end{equation}
with $\kone> -1/\bmax$, $\ktwo\geq0$, and $\ktwo\leq \kone+1/\bmax$.
\end{theorem}
\noindent The proof of Theorem~\ref{thm:nonperv} appears in the Appendix.

Note that if $\kone=0$, then for any $\ktwo\geq0$, the above corresponds to simple scaled marginal-cost tolls.
In this case, the coefficient constraints reduce to $\ktwo\in[0,1/\bmax]$; that is, Theorem~\ref{thm:nonperv} states that scaled marginal-cost taxes have a perversity of $1$ if and only if they are scaled \emph{conservatively,} i.e., they are no larger than would be required to induce optimal flows for a homogeneous population of high sensitivity $\bmax$.

Furthermore, note that all of the results in this paper which are stated for {parallel} networks also hold for networks composed of several parallel networks in series.
\vspace{-2mm}

\subsection{The price of anarchy of non-perverse tolls}

Having shown that parallel networks do admit non-perverse taxation mechanisms, we now ask how effective those mechanisms are in reducing worst-case congestion.
Simply because taxes have a perversity index of $1$ does not immediately imply that their associated PoA is small; nonetheless, we show that generalized marginal-cost tolls can provide modest reductions of worst-case congestion.

\begin{theorem}\label{thm:poaSL}
Let $\gee^{\rm p}_d$ denote the class of all parallel networks with polynomial latency functions of degree at most $d\geq1$.
For any $\bmax\geq\bmin\geq0$, levy the generalized marginal-cost taxation mechanism $\tgmc$ as defined in~\eqref{eq:nonperv} with coefficients $\kone\geq-1/\bmax$, $\ktwo\geq0$, and $\ktwo\leq \kone+1/\bmax$.
Let $\beta_{\kone,\ktwo}\triangleq\frac{\ktwo\bmin}{1+\kone\bmin}\in[0,1]$.
Then the price of anarchy associated with these tolls is
\begin{equation} \label{eq:poaSL}
\poa\left(\gee^{\rm p}_d,\tgmc\right) = \frac{1}{1+d\beta_{\kone,\ktwo}-d\left(\frac{1+d\beta_{\kone,\ktwo}}{1+d}\right)^{\frac{d+1}{d}}}.
\end{equation}
\end{theorem}

We provide a new proof of Theorem~\ref{thm:poaSL} in the Appendix that relies on our arguments for Theorem~\ref{thm:nonperv}, but note that it is also a consequence of~\cite[Theorem 7.1]{Chen2014}, which gives the price of anarchy associated with heterogeneous, partially-altruistic populations.
Though the two proofs are substantially different, they share the high-level idea that on parallel networks, increasing the fraction of players that are merely delay-averse always leads to worse congestion.
This implies (in our model) that the price of anarchy is realized by a homogeneous population with sensitivity equal to $\bmin$, and the expression in~\eqref{eq:poaSL} due to~\cite{Meir2015b} applies immediately.

Furthermore, this yields the following simple characterization of the PoA-minimizing coefficients $\kone$ and $\ktwo$.
Before stating the result, we point out that worst-case performance guarantees provided by a taxation mechanism can often be improved by increasing all edge tolls appropriately (see detailed discussion in~\cite{Brown2017d}).
In order to make meaningful statements about congestion-minimizing tolls, it is useful to parameterize tolls by a stylized upper-bound; the parameter $\Kmax>0$ plays this role in the following result.
Thus, Corollary~\ref{cor:poaSL} solves the following optimization problem:
\begin{equation}\label{eq:poaMin}
(\kone^*,\ktwo^*)\in\arginf_{\kone,\ktwo\leq\Kmax} \poa\left(\gee^{\rm p}_d,\tgmc\right).
\end{equation}

\begin{corollary} \label{cor:poaSL}
For any $d\geq1$ and taxation coefficient upper bound $\Kmax>0$, the price of anarchy in~\eqref{eq:poaMin} due to bounded $\tgmc$ is minimized by setting $\ktwo^*=\Kmax$ and $\kone^*=\Kmax-1/\bmax$.
If $\bmax=+\infty$, then this simplifies to $\kone^*=\ktwo^*=\Kmax$.
\end{corollary}

\begin{proof}
The PoA expression in~\eqref{eq:poaSL} is decreasing in $\beta_{\kone,\ktwo}$, which for any $\bmin,\bmax$ and fixed $\Kmax$ is maximized by saturating the bounds $\kone\geq\ktwo-1/\bmax$ and $\ktwo\leq\Kmax$.
\end{proof}

The price of anarchy due to tolls as in Corollary~\ref{cor:poaSL} is plotted for several values of $d$ in Figure~\ref{fig:nonpervPoa}.
Note that even when $d$ is unbounded (the dotted red curve in Figure~\ref{fig:nonpervPoa}), the price of anarchy is bounded whenever $\bmin/\bmax>0$.

In the special case of $\bmax=+\infty$ (that is, no upper bound on sensitivity is known), the PoA-minimizing coefficients in Corollary~\ref{cor:poaSL} reduce to $\kone=\ktwo=\Kmax$.
Note that this is identically the universal taxation mechanism from~\cite{Brown2017d}, which was developed to serve an entirely different purpose of optimizing the price of anarchy in the large-$\Kmax$ limit.

\begin{figure}
\centerline{\includegraphics[scale=0.27]{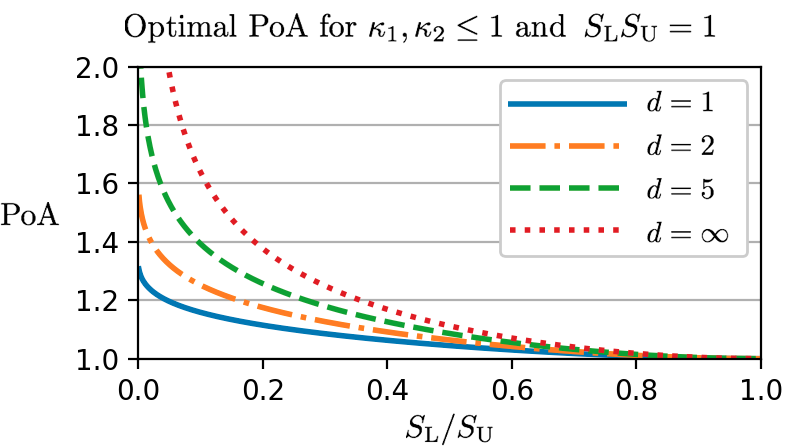}}
\vspace{-.8em}
\caption{The optimal price of anarchy achievable using non-perverse network-agnostic tolls, where $d$ indicates the largest degree of polynomial allowed in the considered latency functions.
These values are plotted using the machinery of Theorem~\ref{thm:poaSL}.
The PoA is plotted with respect to $\bmin/\bmax$, which can serve as a proxy for the variance of the price sensitivities in the user population.
On the far left, the price of anarchy resolves to the un-tolled value; on the right, the price of anarchy is $1$.
Our result continuously bridges the space in between.}
\label{fig:nonpervPoa}
\vspace{-4mm}
\end{figure}


\section{Optimality maximizes perversity} \label{sec:poa} 

In this section we initiate a study on the tradeoff between optimizing for the perversity index and optimizing for the PoA, and show that the choice of coefficients for $\tmech(\kone,\ktwo)$ which minimize the price of anarchy for homogeneous populations on general networks necessarily has a perversity index equal to the price of anarchy.
That is, minimizing the price of anarchy leads to maximizing the perversity index.

\begin{theorem} \label{thm:opt}
Let $\gee_d$ denote the class of all networks with polynomial latency functions of degree at most $d\geq1$.
For homogeneous populations, for any $0\leq\bmin<\bmax$ and taxation coefficient upper bound $\Kmax$, the PoA-minimizing toll scalars
\begin{equation}
(\kone^*,\ktwo^*) \triangleq \arginf_{\kone,\ktwo\leq\Kmax} \poa^{\rm hm}\left( \gee_d, \tmech(\kone,\ktwo) \right)
\end{equation}
satisfy the following:
\begin{align}
\poa^{\rm hm}\left(\gee_d,\tmech(\kone^*,\ktwo^*)\right) = \perv^{\rm hm}\left(\gee_d,\tmech(\kone^*,\ktwo^*)\right) >1. \label{eq:tradeoff}
\end{align}
\end{theorem}

That is, the tolls which minimize the price of anarchy render the perversity index exactly equal to the price of anarchy -- which necessarily must be greater than $1$ (see~\cite{Brown2017d}).
The proof of Theorem~\ref{thm:opt} appears in the Appendix.

\vs\vs
\section{Conclusion}

This paper represents an initial step towards understanding the possible negative consequences of influencing social behavior in network routing situations.
		Note that this paper largely considers only network-agnostic taxation mechanisms, but that in practice many social planners would have access to more information.
		Future work will focus on this closely and attempt to understand how obtaining information regarding the structure of the network can help reduce the risk of unintended consequences.

\bibliographystyle{ieeetr}
\bibliography{../library/library}

\section*{APPENDIX: Proofs}

\noindent\emph{Proof of Lemma~\ref{lem:necessary}:}
We shall exhibit example networks on which various tolls are perverse, thus eliminating all but tolls of the form in~\eqref{eq:nonperv}.
Note that the entire proof of Lemma~\ref{lem:necessary} is completed for homogeneous user populations; this is because non-perversity for homogeneous populations is a necessary condition for non-perversity for heterogeneous populations.
First, consider the network in Figure~\ref{fig:networks}(a).

This network has two paths in parallel; the first path is a pair of edges in series with arbitrary latency functions $\ell_1$ and $\ell_2$, the second path consists of a single edge with latency function $\ell_3$ satisfying $\ell_1+\ell_2=\ell_3$.
For any such network, any nominal Nash flow $\fnash$ is optimal; thus, a non-perverse taxation mechanism $\tmech$ would need to incentivize a flow $f^\tmech$ that satisfies $f^\tmech = \fnash = \fopt = (r/2,r/2)$ by charging tolls satisfying $\tau_1(r/2)+\tau_2(r/2) = \tau_3(r/2)$.
That is, $\tmech$ is additive, or $\tmech(\ell_1)+\tmech(\ell_2) = \tmech(\ell_1+\ell_2)$.
Note that this also implies that $\tmech(0)=0$, since any  function $\ell_1$ can be written $\ell_1+0$.

Next we show that $\tmech(\ell)$ is constant when $\ell$ is constant.
Consider the network in Figure~\ref{fig:networks}(a) when $\ell_1(f_1)=\ell_3(f_3)=b>0$ and $\ell_2(f_2)=\epsilon>0$.
It is clear that the unique Nash and optimal flows both route all traffic on the lower path; i.e., for all $r>0$, $\fnash_3=\fopt_3=r$ and $\fnash_1=\fopt_1=0$.
Writing $\tmech(b)(\cdot)$ as the tolling function assigned to $\ell_1(f_1)=b$ by $\tmech$, it follows that $b+T(b)(r) \leq b+\epsilon + T(b)(0)+T(\epsilon)(0)$, or
\begin{equation}
T(b)(r)-T(b)(0) \leq \epsilon + T(\epsilon)(0). \label{eq:const1}
\end{equation}
Requiring that the map $x\mapsto T(x)(0)$ is Lebesgue-measurable, we obtain that it must be continuous.
Recall that $T(0)(f)=0$, so taking~\eqref{eq:const1} in the limit as $\epsilon\to0$ we obtain that for all $r$ and $b$, $T(b)(r)-T(b)(0) \leq0$, or $T(b)(\cdot)$ is nonincreasing in flow.
An opposite argument shows that $\tmech(b)(\cdot)$ must be a nondecreasing (and thus constant) function of flow.
Because $\tmech$ is a continuous mapping from $\mathbb{R}$ to $\mathbb{R}$ for constant latency functions, its additivity implies linearity: 
$\tmech(b)(f) = \kone b$.
Finally, for all $\epsilon>0$, it must always be true for any possible agent sensitivities $s\in[\bmin,\bmax]$ that $(1+\kone s)b < (1+\kone s) (b+\epsilon)$, or that $\kone>-1/\bmax$.


Next, we show that degree-$d$ monomial latency functions must be assigned degree-$d$ tolling functions.
The network in Figure~\ref{fig:networks}(b) has two edges in parallel with latency functions $\ell_1(f_1) = \alpha (f_1)^d$ and $\ell_2(f_2) = \lambda\alpha (f_2)^d$, where $\alpha>0$, $\lambda>0$, and $d\geq1$.
For any such network, the unique un-influenced Nash flow is optimal; thus, a non-perverse $\tmech$ would need to induce a flow $f^\tmech$ that satisfies $f^\tmech = \fnash = \fopt$.
It can be shown that for any $r>0$, this flow is
$$f^T_1 = \frac{(\lambda\alpha)^{1/d}r}{(\alpha)^{1/d}+(\lambda\alpha)^{1/d}}, \ \ \ f^T_2= \frac{(\alpha)^{1/d}r}{(\alpha)^{1/d}+(\lambda\alpha)^{1/d}}.$$
Since $f^T$ is a nominal Nash flow, $\ell_1(f^T_1) = \ell_2(f^T_2)$; and $f^T=\fopt$ implies that $\tau_1(f^T_1) = \tau_2(f^T_2)$.
In the following, let $r = (\alpha)^{1/d}+(\lambda\alpha)^{1/d}$, so for all $\alpha,\lambda$,
$\tau_1\left((\lambda\alpha)^{1/d}\right) = \tau_2\left((\alpha)^{1/d}\right).$

First let $\lambda = 2$, so that $\ell_2(f_2) = 2\ell_1(f_2)$.
Then additivity ensures that $\tau_2(f_2) = 2\tau_1(f_2)$.
That is,
$	\tau_1\left((2\alpha)^{1/d}\right) = 2\tau_1\left((\alpha)^{1/d}\right).$
Since this must hold for any $\alpha$, it implies either that $T(\alpha f^d)(f)\equiv0$, or that $T(\alpha f^d)(f)= \eta_\alpha f^d$ for some $\eta_\alpha>0$ which depends on $\alpha$.

To characterize $\eta_\alpha$, we substitute $f^T$ into $\tau_1(f^T_1) = \tau_2(f^T_2)$ and solve, yielding
	$\eta_\alpha \lambda 						= \eta_{\lambda\alpha}.$ 
That is, $\eta_{\alpha}=K\alpha$ for some $K\geq0$, so $T(\alpha f^d)(f) = K\alpha (f)^d$ for some $K\geq0$.
Note that $K$ may depend on $d$.

\begin{figure}
\vspace{1mm}
\centerline{\includegraphics[scale=0.22]{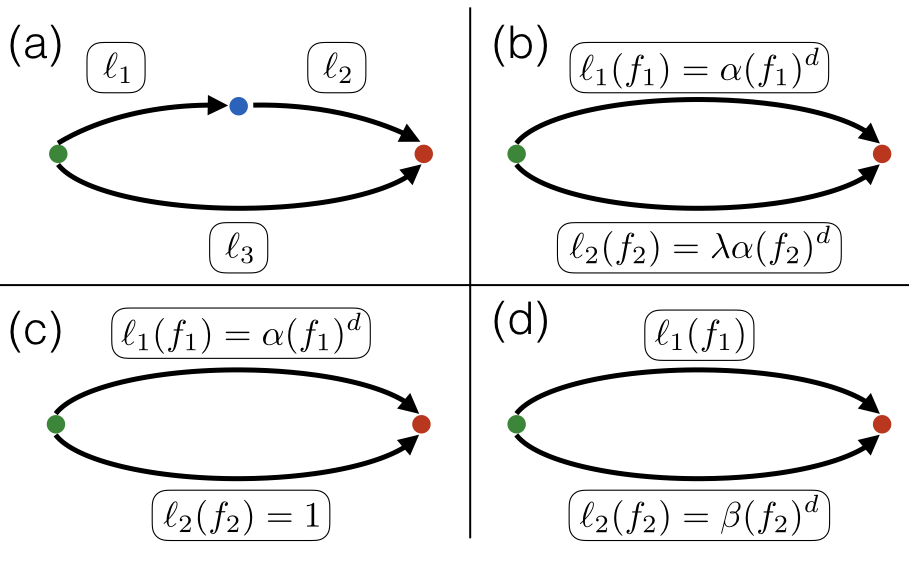}}
\vspace{-1em}
\caption{Example networks used to prove Lemma~\ref{lem:sufficient}.}
\label{fig:networks}
\vspace{-5mm}
\end{figure}

To find $K$, consider Figure~\ref{fig:networks}(c).
This network has $\ell_1(f_1) = \alpha (f_1)^d$ in parallel with a constant latency function $\ell_2(f_2) = 1.$
Here, if $r\leq (1/(\alpha(d+1)))^{1/d}$, the uninfluenced Nash and optimal flow on this network is $(r,0)$.
Thus, $\tau_1(f)$ must be small enough that it does not incentivize \emph{any} user to use edge 2 when $r$ is low.
Precisely, keeping in mind that $\tau_1(f_1)= K\alpha (f_1)^d$, we require that for all sensitivities $s\in[\bmin,\bmax]$,
$\alpha (f_1)^d + sK\alpha(f_1)^d \leq 1 + s\kone,$
or, substituting the appropriate $f_1$, that
\vs\vs\vs
\begin{equation}
	\alpha(1+sK)\left(\left(\frac{1}{\alpha(d+1)}\right)^{1/d}\right)^{d} \leq 1 + s\kone.
\end{equation}
This implies that $sK\leq s\kone d + s\kone + d$.
This simplifies nicely if we write $K=\kone+\ktwo d$ (where $\ktwo\in\mathbb{R}$), in which case it follows that $\ktwo\leq\kone+1/s$ for all $d$ and $s$, or that $\ktwo\leq\kone + 1/\bmax$. 
Writing $\tau_1(f_1)$ in terms of $\kone$ and $\ktwo$ gives the nice decomposition in terms of latency function $\ell_1$ and marginal-cost function $f_1\cdot\ell_1'$:
\vs
\begin{equation}
	\tau_1(f_1) = \kone\ell_1(f_1) + \ktwo f_1\cdot\ell'_1(f_1). \nonumber
\vs
\end{equation}

Finally, consider the network in Figure~\ref{fig:networks}(d).
This network has some arbitrary admissible latency function $\ell_1$ on edge 1 and a latency function $\ell_2(f_2)=\beta (f_2)^d$ on edge 2.
We will choose $\ell_2$ such that the optimal and Nash flows coincide on this network for some $r>1$ when $f_2=1$.

First, assume that $\ell_1(0)=0$.
Let $\beta = \ell_1(f_1)$ and $d = f_1\ell_1'(f_1)/\ell_1(f_1)\geq1$.
Then $\ell_1(f_1) = \ell_2(1)$ and $f_1\ell_1'(f_1) = \ell_2'(1)$; i.e., both the latencies and the marginal costs of the edges are equal, which means that $(f_1,1)$ is both a Nash and an optimal flow.
Since $\ell_2$ is a monomial, we can write its tolling function as $\tau_2(f_2) = \kone \beta(f_2)^d + \ktwo d\beta(f_2)^d$, where $\ktwo\leq\kone+1/\bmax$.
Using this, we can simply derive the first-link tolling function $\tmech({\ell}_1)(f_1)$ using the following:
\vs
\begin{align}
\ell_1(f_1) + s\tmech({\ell}_1)(f_1)	&= \beta (f_2)^d + s\left(\kone\beta + \ktwo d\beta\right) (f_2)^d. \nonumber 	
\vs		
\end{align}
Substituting the definitions of $\beta$ and $d$ and canceling similar terms, we obtain that $T(\ell_1)$ satisfies~\eqref{eq:nonperv} as desired.

Finally, if $\ell_1(0)>0$, define $\tilde{\ell}_1(f_1)\triangleq \ell_1(f_1)-\ell_1(0)$, and let $\ell_2(f_2) = \beta(f_2)^d + \ell_1(0)$ with $\beta=\tilde{\ell}_1(f_1)$ and $d = f_1\ell_1'(f_1)/\tilde{\ell}_1(f_1)\geq1$; then the additivity of $\tmech$ ensures that $T(\ell_1) = T(\tilde{\ell}_1) + \kone(\ell_1(0))$ and $T(\ell_2) = T(\beta(f_2)^d) + \kone(\ell_1(0))$, and the proof proceeds as above.
\hfill\qed

\vspace{-2mm}
\subsection*{Proofs for Theorem~\ref{thm:nonperv}}
\setcounter{theorem}{2}
Next, Lemma~\ref{lem:sufficient} shows that Nash flows on parallel networks behave very nicely under the influence of~$\tgmc$.
Specifically, Lemma~\ref{lem:sufficient} proves that the worst-case total latency on a parallel network with~$\tgmc$ is realized by a low-sensitivity homogeneous population.

\begin{lemma} \label{lem:sufficient}
Let $s^{\rm L}$ denote a homogeneous population in which every user has sensitivity $\bmin\geq0$, and denote by $\tgmc$ a taxation mechanism satisfying the conditions of Lemma~\ref{lem:necessary}.
For any parallel network $G\in\gee^{\rm p}$ and heterogeneous population $s$ in which every user has a sensitivity no less than $\bmin$,
\vs
\begin{equation}
\Lnash\left(G,s^{\rm L},\tgmc\right)\geq \Lnash\left(G,s,\tgmc\right).
\vs\vs\vs\vs
\end{equation}
\end{lemma}
\vspace{2mm}

\begin{proof}
For every user $x$, $\tgmc$ induces cost functions of the form
\vs\vs\vs\vs\vs
\begin{equation}
J^x_e(f_e) = (1+s_x\kone)\ell_e(f_e) + s_x\ktwo f_e \ell'_e(f_e).
\vs
\end{equation}
Since we can scale these costs functions by any user-specific positive scalar without changing the underlying Nash flows, these cost functions are equivalent to the following:
\vs
\begin{equation}
J^x_e(f_e) = \ell_e(f_e) + \frac{s_x\ktwo}{1+s_x\kone} f_e \ell'_e(f_e). \label{eq:thisone}
\vs
\end{equation}
Given the conditions $\ktwo\geq0$ and $\kone\geq \ktwo -1/\bmax$, the expression $\frac{s_x\ktwo}{1+s_x\kone}\in[0,1]$ and is monotone increasing in $s_x$.
Thus, analysis can be simplified with abuse of notation by assuming that cost functions are simply given by the following: \vs \vs
\begin{equation}
J^x_e(f_e) = \ell_e(f_e) + s_x  f_e \ell'_e(f_e), \label{eq:basiccost}
\vs
\end{equation}
where $s_x\in[0,1]$ can be viewed as a synthetic sensitivity.

For convenience, we write $\ell^*_e(f_e)\triangleq f_e\ell_e'(f_e)$.
When describing the cost experienced by a particular agent whose sensitivity is $s\in[0,1]$, we write 
\begin{equation}
\ell_e^s(f_e) \triangleq \ell_e(f_e) + s \ell^*_e(f_e), \label{eq:ells}
\end{equation}
and we write $\ell_e^{\rm mc}(f_e) \triangleq \ell_e^1(f_e)$ to denote the marginal-cost function associated with edge $e$.
The following claim gives important information about the structure of Nash flows induced by~$\tgmc$.

\begin{claim} \label{claim:ordering}
If $\fnash$ is a Nash flow on $G\in\gee^{\rm p}$ for population $s$ under the influence of $\tgmc$, let $p_i$ and $p_j$ be any two paths for which $\fnash_i>0$ and $\fnash_j>0$, and with some user $x$ selecting $p_i$ and user $y$ selecting $p_j$.
Then if $s_x\leq s_y$, we have the following:
\begin{enumerate}
\item \label{item:latOrder} $\ell_i(\fnash_i) \leq \ell_j(\fnash_j)$,
\item \label{item:mcOrder} $\ell^{\rm mc}_i(\fnash_i) \geq \ell^{\rm mc}_j(\fnash_j)$.
\end{enumerate}
\end{claim}

\begin{proof}
Let $\fnash$ be a Nash flow and let $p_i$ and $p_j$ be paths such that $\fnash_i>0$ and $\fnash_j>0$.
Because this is a Nash flow, any agent $y$ using path $p_j$ experiences a (weakly) lower cost than she would on path $p_i$, or
\vs\vs\vs
\begin{equation}
\ell_i(f_i)+ s_y\ell_i^*(f_i) \geq \ell_j(f_{j}) + s_y\ell_{j}^*(f_{j}). \label{eq:1stNash} \vs
\end{equation}
That is,\vs
\begin{equation}
s_y\left(\ell^*_i(f_i)-\ell^*_{j}(f_{j})\right) \geq {\ell}_{j}(f_{j}) - {\ell}_{i}(f_i). \label{eq:2ndNash}
\end{equation}

In the same Nash flow, consider some user $x$ using path $p_i$.
For this user, a similar argument shows that
\begin{equation}
s_x\left(\ell^*_i(f_i)-\ell^*_{j}(f_{j})\right) \leq {\ell}_{j}(f_{j}) - {\ell}_{i}(f_i). \label{eq:3rdNash}
\end{equation}
Combining~\eqref{eq:2ndNash} and~\eqref{eq:3rdNash} yields
\vs
\begin{equation}
0\leq(s_y-s_x)\left(\ell_i^*(f_i)-\ell_{j}^*(f_{j})\right),
\vs
\end{equation}
meaning that $s_y\geq s_x$ is equivalent to $\ell_i^*(f_i)\geq \ell_{j}^*(f_{j})$; inequalities~\eqref{eq:2ndNash} and~\eqref{eq:3rdNash} further imply $\ell_i(f_i)\leq\ell_j(f_j)$, proving item~\eqref{item:latOrder}.
That is, in every Nash equilibrium, lower-sensitivity users pay higher tolls and experience lower latency%
\footnote{
In the language of~\cite{Cole2003}, in our setting, every Nash flow is \emph{canonical.}
}%
.

Henceforth, we index the paths such that $\ell_i^*(\fnash_i)\geq\ell_{i+1}^*(\fnash_{i+1})$.
For each $i\in\{1,\dots,n-1\}$, if $\ell_i^*(\fnash_i)>\ell_{i+1}^*(\fnash_{i+1})$, we define $s_i$ as the number satisfying
\vs
\begin{equation}
\ell_i(f_i) + s_i \ell^*_i(f_i) = \ell_{i+1}(f_{i+1}) + s_i \ell^*_{i+1}(f_{i+1}).  \label{eq:4thNash}
\vs
\end{equation}
When $\ell_i^*(\fnash_i)=\ell_{i+1}^*(\fnash_{i+1})$, let $s_i=s_{i-1}$ (or $0$ if $i=1$).
Then a user with sensitivity $s_i$ (weakly) prefers path $i$ in $\fnash$; thus, each $s_i\leq s_{i+1}$. 
Finally, it follows from $\ell^*_i(f_i)\geq\ell^*_{i+1}(f_{i+1})$ and~\eqref{eq:4thNash} that for any $s_i<1$, we have
$
\ell_i(f_i) + \ell^*_i(f_i) \geq \ell_{i+1}(f_{i+1}) + \ell^*_{i+1}(f_{i+1}),
$
proving item~\eqref{item:mcOrder}.
\end{proof}

The basic proof approach is to exploit this ordering of marginal costs, and show that reducing agents' sensitivities (thereby making the population ``more homogeneous") shifts agents from low marginal-cost paths to high marginal-cost paths, increasing the total latency.
Formally, we define a mapping $\Sigma:[0,1]\times\see\rightarrow \see$.
For any starting population $s^0$ and any $\alpha$, we will define $\Sigma\left(\alpha;s^0\right)$ as a right-shift of $s^0$ by $\alpha$ units.
The sensitivity of user $x$ in population $\Sigma(\alpha,s^0)$ is given by
\vs\vs
\begin{equation}
\Sigma(\alpha,s^0)_x= \left\{
\begin{array}{ll}
s_0(0)		& \mbox{if } x \leq \alpha \\
s_0(x-\alpha) 	& \mbox{if } x > \alpha.
\end{array}
\right.
\vs
\end{equation}
Because $s$ is defined to be an increasing function, this is equivalent to converting a mass of $\alpha$ of the most-sensitive users to a mass $\alpha$ of the least-sensitive users.

Claim~\ref{claim:ordering} allows us to assume without loss of generality that any user population $s$ has a finite number of sensitivity types; to see this, simply note that if users with distinct sensitivities are using the same path in a Nash flow, one sensitivity may be exchanged for the other without perturbing either agent's preferences.
To be precise, given a Nash flow $\fnash$, (indexing the paths so that $\ell_i(\fnash_i)\leq\ell_{i+1}(\fnash_{i+1})$), we will assume for each path $p_i\in\pset\setminus p_1$, each user has the \emph{minimally-indifferent} sensitivity $s_i$ described in~\eqref{eq:4thNash} in the proof of Claim~\ref{claim:ordering}.

For notational brevity, we will typically write $\fnash(\alpha)$ to represent a Nash flow associated with population $\Sigma(\alpha;s_0)$.
Our central goal will be to characterize the effect of marginal increases in $\alpha$ on the Nash flow.
We express this marginal effect as $\frac{\partial}{\partial\alpha}\fnash(\alpha)$. 

The following definition will be helpful in the proof:

\begin{definition} \label{def:sc}
In a Nash flow $\fnash$, 
paths $p_i$ and $p_j$ with $i<j$ are said to be \emph{strategically coupled} if $s_i$ satisfies $\ell_i^{s_i}(\fnash_i) = \ell_{j}^{s_i}(\fnash_{j})$. That is, agents on the lower-index path are indifferent between the two paths. We write $\pset_i(\fnash)$ to denote the set of paths that are strategically coupled to path $p_i$ in $\fnash$.\footnote{When clear from context, we write $\pset_i(\fnash)$ simply as $\pset_i$.}
\end{definition}

First, we show that the primary effect of an increase in $\alpha$ is to shift traffic from $\pset_n$ to $\pset_1$.
\begin{claim} \label{claim:p1pn}
For every path $p_i\in\pset_1$, $\dda\fnash_i(\alpha)\geq 0$. For every path $p_j\in\pset_n$, $\dda\fnash_j(\alpha)\leq 0$.
\end{claim}

\begin{proof}
Let $s_1$ denote the sensitivity of agents using $p_1$ in $\fnash$.
Increasing $\alpha$ changes the sensitivity of a small fraction of high-sensitivity users to $s_1$.
By Definition~\ref{def:sc} and Claim~\ref{claim:ordering}, these users strictly prefer the paths in $\pset_1$ to any other paths, so a marginal increase in $\alpha$ induces a marginal increase in flow on $\pset_1$.
That is, at least one path in $p_i \in \pset_1$ has $\dda\fnash_i(\alpha)>0$.
An implication of Claim~\ref{claim:ordering} is that all paths in $\pset_1$ have strictly flow-varying cost functions, so an increase on flow on $p_i$ induces an increase in flow on all paths in $\pset_1$, proving the first statement.

Next, let $s_n$ denote the sensitivity of agents using $p_n$ in $\fnash$; Definition~\ref{def:sc} and Claim~\ref{claim:ordering} shows that these agents weakly prefer $\pset_n$.
Increasing $\alpha$ shifts some of these users to $\pset_1$, so at least one path in $p_i \in \pset_n$ has $\dda\fnash_i(\alpha)<0$.
If $\pset_n$ contains a path with a constant latency function, then this is the path which the flow leaves; otherwise, the flow would deviate to a non-Nash flow.
On the other hand, if all paths in $\pset_n$ are strictly flow-varying, then every path flow in $\pset_n$ must decrease, proving the second statement.
\end{proof}

\begin{claim} \label{claim:deltaf}
For any $\alpha$, if $p_j \notin \pset_1(\alpha)$ and $p_j \notin \pset_n(\alpha)$, it holds that $\dda\fnash_j(\alpha) = 0$.
That is, $\gamma := \sum_{i\in\pset_1} \frac{\partial}{\partial\alpha}\fnash_i(\alpha) = - \sum_{i\in\pset_n} \frac{\partial}{\partial\alpha}\fnash_i(\alpha).$
\end{claim}

\begin{proof}
%
%
First, let $p_i$ be the lowest-index path such that $p_i\notin\pset_1$ (that is, $p_{i-1}\in\pset_1$).
Definition~\ref{def:sc} means that for any $p_j \in \pset_1$,
$\ell_j^{s_j}(f_j) < \ell_{i}^{s_j}(f_i)$. 
Since the inequality is strict, the fact from Claim~\ref{claim:p1pn} that $\dda\fnash_j(\alpha)\geq0$ means that marginally no agent on $\pset_1$ will switch to $p_i$.

However, since $\fnash(\alpha)$ is a Nash flow, it is true that
$\ell_j^{s_i}(f_j) \geq \ell_{i}^{s_i}(f_i).$ 
Here, $\dda\fnash_j(\alpha)\geq0$ implies that $\ell_j^{s_i}(f_j)$ can only increase, so no agent on $p_i$ will be incentivized to switch to any path in $\pset_1$.
Thus, the flow on $p_i$ is not influenced by the changes in flow on any lower-index path; if its flow changes, the influence must come from a higher-index path.

Now, let $p_i$ be the highest-index path such that $p_i\notin\pset_n$ (that is, $p_{i+1}\in\pset_n$).
Definition~\ref{def:sc} means that for any $p_j \in \pset_n$,
$\ell_i^{s_i}(f_i) < \ell_{j}^{s_i}(f_j).$ 
Since the inequality is strict, the fact that $\dda\fnash_j(\alpha)\leq0$ means that (marginally) no agent on $p_i$ will be incentivized to switch to any path in $\pset_n$.
However, since $\fnash(\alpha)$ is a Nash flow, it is true that
$\ell_j^{s_j}(f_j) \leq \ell_{i}^{s_j}(f_i). $ 
Here, $\dda\fnash_j(\alpha)\leq0$ implies that $\ell_j^{s_j}(f_j)$ can only decrease, so no agent on any path in $\pset_n$ will be incentivized to switch to $p_i$.
Thus, the flow on $p_i$ is not influenced by the changes in flow on any higher-index path.

This argument may then be repeated with all remaining paths that are not in $\pset_1$ or $\pset_n$ to show that the only path flows that may change in response to $\alpha$ are those in $\pset_1$ and $\pset_n$, obtaining the proof.
\end{proof}

\vspace{2mm}
\noindent\emph{Proof of Lemma~\ref{lem:sufficient}:}
We can now quantify the effect of an increase in $\alpha$ on total latency.
By definition, a marginal increase in $\alpha$ corresponds to changing the sensitivity of the highest-sensitivity user to that of the lowest-sensitivity user.
The Nash flow has been represented in such a way that this in turn causes that user to strictly prefer $\pset_1$ and thus switch from $\pset_n$ to $\pset_1$.
In the following, $\nabla_f{\L}(f)$ represents the gradient vector of $\L$ with respect to flow $f$ given by $\{\ell_p^{\rm mc}\}_{p\in\pset}$, which by Claim~\ref{claim:ordering} is ordered descending. Let $p_{j}$ be the highest-index path in $\pset_1$, and $p_{k}$ be the lowest-index path in $\pset_n$.
\begin{align}
	\dda \L\left( \fnash(\alpha)\right) 	& = \nabla_f \L\left(\fnash(\alpha)\right) \cdot \dda \fnash(\alpha) \nonumber \\
								& = \sum\limits_{i\in \pset_1\cup\pset_n} \ell_i^{\rm mc}\left(\fnash_i(\alpha)\right) \dda\fnash_i(\alpha) \nonumber\\
								& \geq \gamma\left[\ell_j^{\rm mc}\left(\fnash_j(\alpha)\right)-\ell_k^{\rm mc}\left(\fnash_k(\alpha)\right)\right] \nonumber 
								 \geq 0.
\end{align}
Here, the last line is a consequence of Claims~\ref{claim:p1pn} and~\ref{claim:deltaf}: since $p_j$ is defined as a member of $\pset_1$, we have that $\frac{\partial}{\partial \alpha} f_j^{\rm nf} (\alpha)\geq0$, and since $p_k$ is defined as a member of $\pset_n$, we have that $\frac{\partial}{\partial \alpha} f_k^{\rm nf} (\alpha)\leq0$; the last inequality then follows from the ordering of marginal costs as in Claim~\ref{claim:ordering}.

Since at every Nash flow $\fnash(\alpha)$ it is true that $\dda \L\left( \fnash(\alpha)\right)\geq0$, the definition of $\Sigma\left(\alpha,s_0\right)$ implies that for any initial sensitivity distribution $s_0$,
\begin{equation}
\mathcal{L}\left(\fnash\left(\Sigma\left(1,s_0\right)\right)\right) \geq \mathcal{L}\left(\fnash\left(\Sigma\left(0,s_0\right)\right)\right),
\end{equation}
or that $\Lnash\left(G,s^{\rm L},\tgmc\right)\geq \Lnash\left(G,s,\tgmc\right).$
\end{proof}

\subsubsection*{Proof of Theorem~\ref{thm:nonperv}}
Let $G\in\gee^{\rm p}$ be a parallel network, $s$ be any arbitrary sensitivity distribution, $s^{\rm L}$ be a homogeneous population in which all users have sensitivity $\bmin$, and let taxation mechanism $\tgmc$ satisfy~\eqref{eq:nonperv}.
Lemma~\ref{lem:sufficient} ensures
\vs\vs
\begin{equation}
\Lnash\left(G,s^{\rm L},\tgmc\right)\geq \Lnash\left(G,s,\tgmc\right). \label{eq:proofineq1}
\vs\vs
\end{equation}

Let $s^0$ denote a totally-insensitive homogeneous population; that is, all agents have sensitivity $0$. Note that $s^0$ is itself a low-sensitivity homogeneous population and that $s$ is a population in which all users have sensitivity no less than $0$; thus, we may simply apply Lemma~\ref{lem:sufficient} a second time to obtain
\vs
\begin{equation}
\Lnash\left(G,s^0,\tgmc\right)\geq \Lnash\left(G,s^{\rm L},\tgmc\right). \label{eq:proofineq2}
\vs
\end{equation}
The left-hand side of~\eqref{eq:proofineq2} is simply the un-tolled total latency on $G$, so combining inequalitites~\eqref{eq:proofineq1} and~\eqref{eq:proofineq2}, we obtain
\vs
\begin{equation}
\Lnash(G,\emptyset) \geq \Lnash\left(G,s,\tgmc\right).\vs
\end{equation}
Since $G$ and $s$ were arbitrary, this implies that tolls of the form in~\eqref{eq:nonperv} have a perversity index of $1$ on $\gee^{\rm p}$.
Finally, Lemma~\ref{lem:necessary} shows that these tolls are \emph{necessary} as well.
\hfill\qed

\vspace{2mm}
\noindent\emph{Proof of Theorem~\ref{thm:poaSL}:}
Let $\beta(s,(\kone,\ktwo)):=\frac{\ktwo s}{1+\kone s}$.
The constraints on $\kone,\ktwo$ specified in Theorem~\ref{thm:nonperv} imply that for all $s\in[\bmin,\bmax]$, $\beta(s,(\kone,\ktwo))\in[0,1]$.
Lemma~\ref{lem:sufficient} implies that on any $G\in\gee^{\rm p}$, worst-case routing is achieved by a homogeneous population with $s=\bmin$.
Thus, the price of anarchy for heterogeneous populations is equal to that given by Lemma~\ref{lem:meir} for $\beta\leq1$, proving the theorem.
\hfill\qed

\subsection*{Proofs for Theorem~\ref{thm:opt}}
\setcounter{theorem}{5}

\noindent In the following, we frequently refer to the quantity 
\begin{equation}
\beta(s,(\kone,\ktwo)) := \frac{\ktwo s}{1+\kone s};
\end{equation}
since $\tmech(\kone,\ktwo)$ induces cost functions of~\eqref{eq:thisone}. 

That is, $\beta(s,(\kone,\ktwo))$ indicates users' induced sensitivity to their marginal effect on others: when $\beta=1$, users interpret their marginal latency effect on others correctly.
When $\beta<1$, users are overly delay-sensitive; when $\beta>1$, users are not delay-sensitive enough.
Note that $\tgmc$ corresponds to $\beta\leq1$.

First, we present the following lemma, which is adapted from~\cite[Propositions 7.6, 7.7]{Meir2015b},
giving the price of anarchy for known-sensitivity homogeneous population for fixed $\kone,\ktwo$:
\begin{lemma}[Meir and Parkes, 2015~\cite{Meir2015b}] \label{lem:meir}
Let $\gee_d$ be the class of routing problems with polynomial latency functions of degree no more than $d\geq1$, and let $\gee$ denote the class of all routing problems with polynomial latency functions.
In the following, let $\kappa:=(\kone,\ktwo)$ and let $\beta(s,\kappa):=\frac{\ktwo s}{1+\kone s} \geq 0$.
Then the price of anarchy resulting from $\tmech(\kone,\ktwo)$ for a homogeneous population with sensitivity $s$ is
\begin{equation} \label{eq:poahelper}
\overbar{\poa}\left(\gee_d,s,\kappa\right)\hs =\hs \left\{\hs\hs\hs
\begin{array}{cl}
\frac{1}{1+d\ugly-d\left(\hs\frac{1+d\ugly}{1+d}\hs\right)^{\frac{d+1}{d}}} 	&\hs\hs\hs\mbox{ if } \ugly \leq 1, \vspace{1mm}\\
\uglyinv^{-d}\left(\hs\frac{1+d\ugly}{1+d}\hs\right)^{d+1} 				&\hs\hs\hs\mbox{ if } \ugly > 1.
\end{array}\right.
\end{equation}
\end{lemma}

\noindent Lemma~\ref{lem:meir} follows from~\cite[Propositions 7.6, 7.7]{Meir2015b} for reasons elucidated in the proof of Lemma~\ref{lem:sufficient}.
This allows us to state:

\begin{lemma}\label{lem:pervSU}
Let $\gee^{\rm p}_d$ denote the class of routing problems with parallel networks and polynomial latency functions of degree no more than $d\geq 1$.
Given $\ktwo>0$, when $\kone\in(-1/\bmax,\ktwo-1/\bmax)$, the homogeneous perversity index of $\tmech(\kone,\ktwo)$ is greater than $1$ and equal to the price of anarchy experienced by a population with sensitivity $\bmax$:
\begin{equation}
\perv^{\rm hm}\left(\gee^{\rm p}_d,\tmech(\kone,\ktwo)\right) = \auxpoa\left(\gee_d,\bmax,(\kone,\ktwo)\right).
\end{equation}
\end{lemma}

\vspace{2mm}
\noindent\emph{Proof of Lemma~\ref{lem:pervSU}:}
Let $\beta(s):=\frac{\ktwo s}{1+\kone s}$.
By Lemma~\ref{lem:sufficient}, the perversity index can never be greater than $1$ \emph{due to a population with  $\beta(s)\leq 1$.}
However, $\beta(s)>1$ is increasing in $s$ whenever $\ktwo>0$ and $\kone>-1/\bmax$, implying that the price of anarchy \emph{due to a population with $\beta(s)>1$} is achieved by one with $s=\bmax$.
To prove the lemma, we show that there is a perverse flow for $s=\bmax$ whose perversity equals the corresponding price of anarchy.
Accordingly, let $\beta:=\frac{\ktwo\bmax}{1+\kone\bmax}>1$.
Consider the network in Figure~\ref{fig:networks}(c) with $r=1$ and $\alpha=\frac{\left(\beta(1+d)\right)^d}{(1+d\beta)^{d+1}}$; this network is borrowed from~\cite{Meir2015b}.
The optimal and uninfluenced Nash flow on this network has $f_1=1$, but the tolled Nash flow when $s=\bmax$ has larger total latency which achieves the $\beta>1$ PoA bound given in~\eqref{eq:poahelper}, completing the proof.
\hfill\qed

\vspace{2mm}
\noindent\emph{Proof of Theorem~\ref{thm:opt}:}
The restriction to homogeneous populations here allows us to apply Lemma~\ref{lem:meir} directly, and leverage its monotonicity properties to obtain the result.
Consider the expressions given by~\eqref{eq:poahelper} as a function of $\beta(s,\kappa)$.
The price of anarchy as a function of $\beta$ is bowl-shaped:
when $\beta(s,\kappa)<1$, the price of anarchy is strictly decreasing in $\beta(s,\kappa)$,
when $\beta(s,\kappa)>1$, the price of anarchy is strictly increasing in $\beta(s,\kappa)$, and 
when $\beta(s,\kappa)=1$, the price of anarchy is equal to $1$.
Thus, minimizing the price of anarchy reduces to choosing $\kone$ and $\ktwo$ such that for all $s\in[\bmin,\bmax]$, $\beta(s,\kappa)$ takes values as ``close'' to $1$ as possible, where this closeness is measured by the expressions in~\eqref{eq:poahelper}.
Furthermore, when $\kone>-1/\bmax$ and $\ktwo\geq0$, the monotonicity of $\beta(s,\kappa)$ ensures that the price of anarchy is achieved by an extreme sensitivity population with $s\in\{\bmin,\bmax\}$.

For any fixed $\ktwo>0$ and fixed $s\in[\bmin,\bmax]$, $\beta(s,\kappa)$ is decreasing in $\kone$ whenever $\kone>-1/\bmax$; also, when $\kone=\ktwo-1/s$, we have $\beta(s,\kappa)=1$ (thus, the PoA for that $s$ is $1$).
Combined with the above, this means the price of anarchy is minimized in $\kone$ on the interval $I:=\left(\max\{-1/\bmax,\ktwo-1/\bmin\},\ktwo-1/\bmax\right)$.

Accordingly, for any $\ktwo$, let $\delta:=\Kmax-\ktwo$.
For any $\ktwo>0$ and $\kone\in I$, it holds that \vs\vs
$$\vs\vs\beta(\bmin,(\kone,\ktwo)<\beta(\bmin,(\kone+\delta,\ktwo+\delta)),$$
and\vs\vs
$$\vs\vs\beta(\bmax,(\kone,\ktwo)>\beta(\bmax,(\kone+\delta,\ktwo+\delta)).$$
That is, when $\kone\in I$, the price of anarchy can be decreased by adding $\delta$ to both $\kone$ and $\ktwo$, showing that $\ktwo^*=\Kmax$.

Finally, given that $\ktwo^*=\Kmax$ and $\kone\in I$, note that $\auxpoa(\gee_d,\bmin,(\kone,\ktwo^*))$ is strictly increasing in $\kone$, that~$\auxpoa(\gee_d,\bmax,(\kone,\ktwo^*))$ is strictly decreasing in~$\kone$, and both are continuous for all $\kone$.
This guarantees that the PoA-minimizing $\kone^*$ must lie on the interval %
$\left(-1/\bmax,\Kmax-1/\bmax \right)$, and~\eqref{eq:tradeoff} is an immediate consequence of Lemma~\ref{lem:pervSU}.
\hfill\qed

%

\vspace{-.25in}
\begin{IEEEbiography}[{\includegraphics[width=1in,height=1.25in,clip,keepaspectratio]{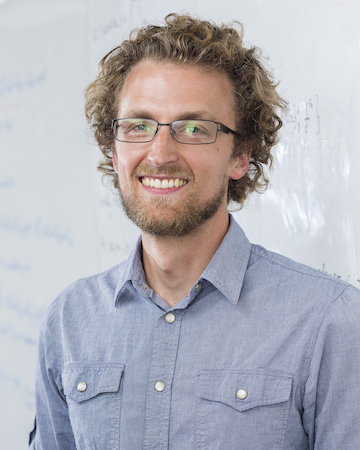}}]{Philip N. Brown}
is an Assistant Professor in the Department of Computer Science at the University of Colorado, Colorado Springs.
Philip received the Bachelor of Science in Electrical Engineering in 2007 from Georgia Tech, after which he spent several years designing control systems and process technology for the biodiesel industry.
He received the Master of Science in Electrical Engineering in 2015 from the University of Colorado at Boulder under the supervision of Jason R. Marden, where he was a recipient of the University of Colorado Chancellor's Fellowship.
He received the PhD in Electrical and Computer Engineering from the University of California, Santa Barbara under the supervision of Jason R. Marden.
He was finalist for the Best Student Paper Award at the 2016 and 2017 IEEE Conferences on Decision and Control, and received the 2018 CCDC Best PhD Thesis Award from UCSB.
Philip is interested in the interactions between engineered and social systems.
\end{IEEEbiography}


\vfill
\vspace{-1in}
\begin{IEEEbiography}[{\includegraphics[width=1in,height=1.25in,clip,keepaspectratio]{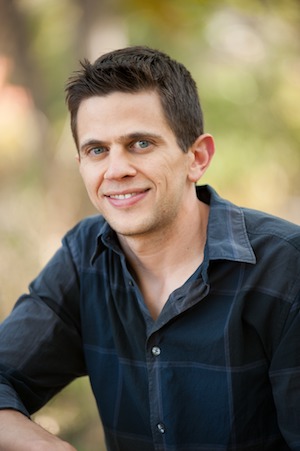}}]{Jason R. Marden}
is an Associate Professor in the Department of Electrical and Computer Engineering at the University of California, Santa Barbara.
Jason received the Bachelor of Science in Mechanical Engineering in 2001 from UCLA, and the PhD in Mechanical Engineering in 2007, also from UCLA, under the supervision of Jeff S. Shamma, where he was awarded the Outstanding Graduating PhD Student in Mechanical Engineering.
After graduating from UCLA, he served as a junior fellow in the Social and Information Sciences Laboratory at the California Institute of Technology until 2010, and then as an Assistant Professor at the University of Colorado until 2015.
Jason is a recipient of an ONR Young Investigator Award (2015), NSF Career Award (2014), AFOSR Young Investigator Award (2012), SIAM CST Best Sicon Paper Award (2015), and the American Automatic Control Council Donald P. Eckman Award (2012).
Jason's research interests focus on game theoretic methods for the control of distributed multiagent systems.
\end{IEEEbiography}






\end{document}

%% file: philipMacros.tex
\newcommand{\ugly}{\beta(s,\kappa)}
\newcommand{\uglyinv}{\beta(s,\kappa)}

\DeclareMathOperator*{\arginf}{arg\,inf}

\newcommand{\perv}{{\rm PI}}
\newcommand{\tgmc}{\tmech^{\rm gmc}}
\newcommand{\fperv}{f^{\rm perverse}}

\newcommand{\dda}{\frac{\partial}{\partial\alpha}}

\newcommand{\vs}{\vspace{-1mm}}
\newcommand{\hs}[1][0.5]{\hspace{-#1mm}}

\newcommand{\gee}{\mathcal{G}}

\newcommand{\fnash}{\ensuremath{f^{\rm nf}}}

\newcommand{\fopt}{\ensuremath{f^{\rm opt}}}

\newcommand{\Kmax}{M}
\newcommand{\kone}{\kappa_1}
\newcommand{\ktwo}{\kappa_2}

\newcommand{\pset}{\ensuremath{\mathcal{P}}}

\newcommand{\poa}{\ensuremath{{\rm PoA}}}
\newcommand{\auxpoa}{\overbar{\poa}}
\newcommand{\tmc}{\tmech^{\rm mc}}

\newcommand{\tmech}{T}

\newcommand{\bmin}{\ensuremath{S_{\rm L}}}
\newcommand{\bmax}{\ensuremath{S_{\rm U}}}

\newcommand{\see}{\mathscr{S}}

\renewcommand{\L}{\mathcal{L}}
\newcommand{\Lnash}{\ensuremath{\mathcal{L}^{\rm nf}}}
\newcommand{\Lopt}{\ensuremath{\mathcal{L}^*}}

\newcommand{\Lnf}{\mathcal{L}^{\rm nf}}
